\documentclass[journal,twocolumn]{IEEEtran}

\ifCLASSINFOpdf
  
\else
  
\fi

\usepackage{amsmath}
\usepackage{amsthm}
\usepackage{cases}
\usepackage{amsfonts}
\usepackage{cite}
\usepackage{amssymb}
\usepackage{enumerate}
\usepackage{graphicx}
\usepackage{float}
\usepackage{caption}
\usepackage{verbatim}
\usepackage{dblfloatfix}
\usepackage[cmintegrals]{newtxmath}
\usepackage[ruled,linesnumbered]{algorithm2e}
\usepackage{color}
\usepackage{soul}

\newtheorem{theorem}{Theorem}

\newtheorem{definition}{Definition}
\newtheorem{corollary}{Corollary}
\newtheorem{remark}{Remark}
\newtheorem{proposition}{Proposition}

\begin{document}

\title{Lossy Source Coding with Broadcast Side Information}

\author{%
\IEEEauthorblockN{Yiqi Chen{$^*$},
  Holger Boche{$^*$}, Marc Geitz{$^\dagger$}}\\
\IEEEauthorblockA{{$*$} Technical University of Munich,  
80333 Munich, Germany,\;\;\{yiqi.chen, boche\}@tum.de\\
{$\dagger$} T-Labs, Deutsche Telekom AG, Germany, marc.geitz@telekom.de }
}

\maketitle
\thispagestyle{empty}
\pagestyle{empty}

\begin{abstract}
This paper considers the source coding problem with broadcast side information. The side information is sent to two receivers through a noisy broadcast channel. We provide an outer bound of the rate--distortion--bandwidth (RDB) quadruples and achievable RDB quadruples when the helper uses a separation-based scheme. Some special cases with full characterization are also provided. We then compare the separation-based scheme with the uncoded scheme in the quadratic Gaussian case.
\end{abstract}

\ifCLASSOPTIONcaptionsoff
  \newpage
\fi

\section{Introduction}

The performance limitation of transmitting a source through a point-to-point noisy channel has been fully characterized by Shannon's source-channel separation theorem\cite{shannon1948mathematical}. However, such a separation theorem generally does not hold for multi-user joint source-channel coding (JSCC) problems.

A classic yet open problem in multi-user JSCC is broadcasting a single source to multiple receivers.
Recent works on this topic mainly focus on designing digital/analog/hybrid schemes \cite{prabhakaran2011hybrid,nayak2010wyner,gao2011wyner,minero2015unified,tian2013optimality} and characterizing outer bounds on the admissible distortion regions \cite{tian2010approximate,khezeli2015outer,yu2018distortion,reznic2006distortion,khezeli2016source,gohari2008outer}.
Specifically, for transmitting a Gaussian source through a Gaussian broadcast channel, the uncoded scheme is shown to be optimal only when the source bandwidth matches the channel bandwidth, while the optimal strategy for bandwidth mismatch cases remains open. 

This paper considers the multi-user lossy compression of a source $S$ with two receivers and a helper. However, the helper who observes a correlated source $T$ can only communicate to the receivers through a noisy broadcast channel. 
 A similar problem studied in \cite{shamai2002systematic} considers the case where two senders observe the same source $S$ and communicate to the receiver through two separate noisy channels. Sufficient and necessary conditions for the separation-based scheme to be optimal were derived. We extend the model by considering a pair of correlated sources $(S,T)$ and two receivers, and simplify the link from the sender to the receiver to a noiseless link.

When the noisy channel is a degraded broadcast channel, our model also shares some similarities with source coding with degraded side information (SI) \cite{heegard1983capacity,tian2007multistage}. For some special cases, our model in fact reduces to source coding with degraded SI when an uncoded scheme is used. We compare these two problems with a quadratic Gaussian example.\footnote{Throughout this paper, random variables, sample values and their alphabets are denoted by capital, lower and calligraphic letters, e.g., $X,x$ and $\mathcal{X}$. A sequence with length $n$ is denoted by $X^n=(X_1,X_2,...,X_n)$. $\mathbb{E}[\cdot]$ is the expectation.}

\section{definitions}

The sender observes the source $S^n$ and sends a lossy description to both receivers at rate $R$. The helper observes a correlated source $T^n$. It encodes the source into some codeword $X^l$ and inputs it to a degraded broadcast channel $P_{YZ|X}=P_{Y|X}P_{Z|Y}$. The receivers observe channel output $Z^l$ and $Y^l$ and reconstruct $\hat{S}^n_1$ and $\hat{S}^n_2$ with distortion levels $D_1$ and $D_2$, respectively, such that
\begin{align}
    \mathbb{E}[d(S^n,\hat{S}^n_1)] \leq D_1,\;\mathbb{E}[d(S^n,\hat{S}^n_2)] \leq D_2
\end{align}
for some distortion function $d: \mathcal{S}\times \hat{\mathcal{S}} \to [0,+\infty)$. The bandwidth expansion factor is defined by $\rho = l/n$. Furthermore, due to the characteristics of the broadcast channel according to our assumption, we call the receiver who observes $Z^l$ the weak receiver and the other one the strong receiver.

In some of the results, we restrict the helper to use a separation-based coding scheme. That says, the helper first quantizes the source $T^n$ using a lossy source code, and then encodes it using a channel code. The formal definitions of the general joint-source channel coding and separation-based coding schemes are as follows.
\begin{figure}[t]
    \centering
    \includegraphics[scale=0.7]{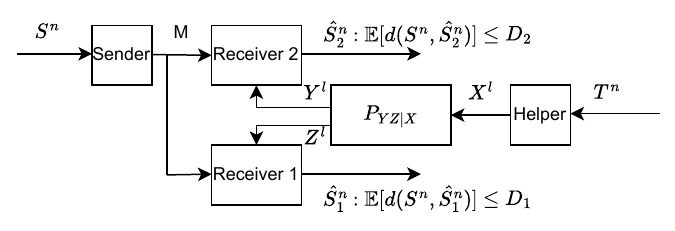}
    \caption{\footnotesize Lossy compression with broadcasted side information}
    \label{fig:enter-label}
\end{figure}
\begin{definition}\label{def: joint source channel code}
    A joint source-channel code with a helper $(\rho,R,f,g)$ includes
    \begin{itemize}
        \item A message set of the source encoder $\mathcal{M}=[1:2^{nR}]$.
        \item A source encoder $f_1:\mathcal{S}^n\to \mathcal{M}.$
        \item A helper encoder $f_2: \mathcal{T}^n \to \mathcal{X}^l$,
        \item A pair of decoders $g_1:\mathcal{M}\times\mathcal{Z}^l \to \hat{S}^n_1,g_2: \mathcal{M}\times\mathcal{Y}^l \to \hat{S}^n_2$.
    \end{itemize}
\end{definition}
As a special case of the joint source-channel code, we define the following separation-based coding scheme:
\begin{definition}\label{def: separation-based scheme for the helper}
    The separation-based scheme for the helper includes
    \begin{itemize}
        \item A concatenated encoder $f_{21}\circ f_{22}$, where
        \begin{align}
            &f_{21}:\mathcal{T}^n \to \mathcal{M}_1 \times \mathcal{M}_2,\;\;f_{22}:\mathcal{M}_1\times\mathcal{M}_2 \to \mathcal{X}^l.
        \end{align}
        \item A pair of concatenated decoders $g_{i1}\circ g_{i2},i=1,2$, where
        \begin{align}
            &g_{11}:\mathcal{Z}^l\to\mathcal{M}_1, \;\;g_{12}:\mathcal{M}\times\mathcal{M}_1  \to \hat{\mathcal{S}}^n_1,\\
            &g_{21}:\mathcal{Y}^l\to \mathcal{M}_2,\;\;g_{22}:\mathcal{M}\times \mathcal{M}_2 \to \hat{\mathcal{S}}^n_2.
        \end{align}
    \end{itemize}
\end{definition}

\begin{definition}
    A source code $(\rho,R,f_1,f_{21}\circ f_{22},g_{11}\circ g_{12},g_{21}\circ g_{22})$ for $S^n$ with a correlated source $T^n$ observed by a separation-based helper includes
    \begin{itemize}
        \item A message set of the source encoder $\mathcal{M}=[1:2^{nR}]$.
        \item A source encoder $f_1:\mathcal{S}^n\to \mathcal{M}.$
        \item A concatenated helper encoder $f_{21}\circ f_{22}$ and a pair of concatenated decoders $g_{i1}\circ g_{i2},i=1,2$.
    \end{itemize}
    For simplicity, we use $(\rho,R,f,g)_S$ to denote the code.
\end{definition}
Let $(\hat{M}_1,\hat{M}_2)$ be the decoding results of $(M_1,M_2).$
\begin{definition}
    A quadruple $(R,D_1,D_2,\rho)$ is achievable if for any $\epsilon >0$ there exists a joint source-channel code $(\rho,R,f,g)$ such that
    \begin{align}
    &l/n \leq \rho  + \epsilon,\\
        &\mathbb{E}[d(S^n,\hat{S}^n_1(M,Z^l))]\leq D_1 + \epsilon,\\
        &\mathbb{E}[d(S^n,\hat{S}^n_2(M,Y^l))]\leq D_2 + \epsilon.
    \end{align}
    The rate-distortion-bandwidth (RDB) region $\mathcal{R}$ is the set of all $(R,D_1,D_2,\rho)$ quadruples that are achievable.
    
    When the helper is restricted to the separation-based scheme, a quadruple $(R,D_1,D_2,\rho)$ is achievable if for any $\epsilon>0$ there exists a code $(\rho,R,f,g)_S$ such that
    \begin{align}
    &l/n \leq \rho  + \epsilon,\\
    &Pr\{(M_1,M_2)\neq (\hat{M}_1,\hat{M}_2)\}\leq\epsilon,\\
    &\mathbb{E}[d(S^n,\hat{S}^n_1(M,\hat{M}_1))]\leq D_1 + \epsilon,\\
        &\mathbb{E}[d(S^n,\hat{S}^n_2(M,\hat{M}_2))]\leq D_2 + \epsilon.
    \end{align}
    The RDB region $\mathcal{R}_{Se}$ is the set of all $(R,D_1,D_2,\rho)$ quadruples that are achievable with a separation-based scheme.
\end{definition}

\section{main results}
The main result characterizes the tradeoff between the sender’s compression rate, the helper’s broadcast bandwidth, and the reconstruction distortions at two receivers when side information is broadcast over a degraded channel. It shows that a layered (separation-based) scheme is optimal  when the helper observes the source directly, yielding simple additive information constraints that fully describe the rate–distortion–bandwidth region.
\begin{theorem}\label{coro: general outer bound degraded case}
    Given a joint distribution $P_{ST}$, for the lossy compression of the source $S$ with a helper who observes $T$ and communicates to the receiver through a degraded BC $P_{YZ|X},$ a quadruple $(R,D_1,D_2,\rho)$ is achievable only if for every distribution $P_{U|T}$, there exists a distribution $P_{\hat{S}_1|SV_1}P_{\hat{S}_2|SV_2}$ and Markov chains $(S,U)-T-V_2-V_1$ and $W-X-(Y,Z)$ such that $R\geq 0 $ and
    \begin{align}
        &R \geq I(\hat{S}_1;U) + I(\hat{S}_1,\hat{S}_2;S|U) - I(U;V_1) - I(V_2;S|U),\\
        &R \geq I(\hat{S}_1;U) - I(U;V_1),\\
        &\rho I(W;Z) \geq I(V_1;U),\\
        &\rho I(W;Z) + \rho I(X;Y|W) \geq I(V_1;U) + I(V_2;T|U),
        \end{align}
        \begin{align}
        &\mathbb{E}[d(S,\hat{S}_1)] \leq D_1, \;\mathbb{E}[d(S,\hat{S}_2)] \leq D_2.
    \end{align}
\end{theorem}
\begin{remark}
    The outer bound uses a remote source argument that was also used in \cite{tian2010approximate,khezeli2015outer,yu2018distortion,reznic2006distortion}. The bound is motivated by the rate loss in the source coding with a helper problem.  A simple example is the lossless source coding with a helper. The optimal sender-helper rate region is the set of $(R,R_h)$ s.t.
\begin{align}
    &R \geq H(S|V)=H(S)-I(S;V),\;R_h \geq I(T;V).
\end{align}
Only $I(S;V)$ bits out of $I(T;V)$ bits of information sent by the helper are useful for the sender. We provide an outer bound for the single receiver case in Appendix \ref{app: single user outer bound} for a better understanding.
\end{remark}
\begin{definition}\label{def: definition of the one helper inner boud region}
    Let $\mathcal{R}^{in}_{Se}$ be the set of quadruple $(R,D_1,D_2,\rho)$ such that there exists a set of finite auxiliary random variables $(U_1,U_2,V_1,V_2,W)$ such that
    \begin{align}
        \label{ine: one helper rate constraint 1}&R > I(U_1;S|V_1) + I(U_2;S|U_1,V_2),\\
        \label{ine: one helper rate constraint 4}&I(V_1;T) \leq \rho I(W;Z),\\
        \label{ine: one helper rate constraint 5}&I(V_1;T) + I(V_2;T|U_1,V_1)\leq \rho I(W;Z)+ \rho I(X;Y|W).
    \end{align}
with the joint distribution $P_{ST}P_{U_1U_2|S}P_{V_2|T}P_{V_1|V_2}P_WP_{X|W}P_{YZ|X}.$
Furthermore, there exist deterministic functions $h_1:\mathcal{U}_1\times\mathcal{V}_1 \to \hat{\mathcal{S}_1}, h_2:\mathcal{U}_1\times\mathcal{U}_2\times\mathcal{V}_2 \to \hat{\mathcal{S}_2}$ such that
\begin{align*}
    &\mathbb{E}[d(S,h_1(U_1,V_1))] \leq D_1,\;\mathbb{E}[d(S,h_2(U_1,U_2,V_2))] \leq D_2.
\end{align*}
\end{definition}

\begin{theorem}\label{the: general one helper inner bound}
    The RDB region of the lossy source coding with broadcasted side information under the separation assumption satisfies $\mathcal{R}^{in}_{Se} \subseteq \mathcal{R}_{Se}.$
    Furthermore, the cardinalities of the auxiliary random variable alphabets satisfy
    \begin{align}
        &|\mathcal{U}_1| \leq |\mathcal{S}|+4, \;|\mathcal{U}_2| \leq |\mathcal{S}|(|\mathcal{S}|+4)+1\\
        &|\mathcal{V}_1| \leq |\mathcal{T}|+3,\;|\mathcal{V}_2| \leq |\mathcal{T}|(|\mathcal{T}|+3)+1,\\
        &|\mathcal{W}| \leq |\mathcal{X}|+1.
    \end{align}
\end{theorem}
\emph{Sketch of achievability:} The message of the sender $S$ consists of two parts: a common part $U_1$ that can be decoded by both receivers and a private part $U_2$ that can only be decoded by the strong receiver. As we assumed that the helper uses a separation-based scheme, it first compresses $T^n$ into a two-layer description. The first layer $V_1$ is generated at rate  $I(V_1;T)$ using a lossy source coding code. Correspondingly, the common part $U_1$ at the sender is generated at rate $I(U_1;S|V_1)$ using a Wyner-Ziv code with $V_1$ the decoder-side information. The second helper layer $V_2$ is generated at rate $I(V_2;T|U_1,V_1)$ using a Wyner-Ziv code with $V_1$ being the common side information and $U_1$ being the decoder-side information. Accordingly, the private message of the sender $U_2$ is generated using a Wyner-Ziv code with rate $I(U_2;S|U_1,V_2)$ with side information $(U_1,V_2).$ The helper then uses a superposition code for the degraded broadcast channel with $V_1$ being the common message and $V_2$ being the private message.

\begin{corollary}\label{coro: T=S separation based scheme optimal case}
    When $S=T$, the quadruple $(R,D_1,D_2,\rho)$ is achievable with a separation-based scheme if and only if
    \begin{align}
        &R + \rho I(W;Z) \geq I(\hat{S}_1;S),\\
        &R + \rho I(W;Z) + \rho I(X;Y|W) \geq I(\hat{S}_1,\hat{S}_2;S)
    \end{align}
    with the joint distribution $P_{S\hat{S}_1\hat{S}_2}P_WP_{X|W}P_{YZ|X}$ such that
    \vspace{-0.1in}\begin{align}
        &\mathbb{E}[d(S,\hat{S}_1)]\leq D_1,\;\mathbb{E}[d(S,\hat{S}_2)]\leq D_2.
    \end{align}
\end{corollary}
The converse is given in Section \ref{app: proof of T=S separation based scheme optimal case}. The achievability follows by constructing lossy descriptions $\{\hat{s}^n_1\},\{\hat{s}^n_2\}$ by the marginal distributions $P_{\hat{S}_1}P_{\hat{S}_2|\hat{S}_1}$ and dividing their indices into two parts, sent by the sender and the helper, respectively. The details are provided in Appendix \ref{app: achievability of T=S separation case}.
In the following, we consider a special case in which the distortion function at the weak receiver is a deterministic distortion function\cite{gamal1982achievable}\cite{timo2014source}.
\begin{definition}
    Given a function $\psi:\mathcal{S}\to \widetilde{\mathcal{S}}$ of the source, a distortion function $d:\mathcal{S}\times\hat{\mathcal{S}}\to [0,+\infty)$ is called a deterministic distortion if 
    \begin{align}
        d(s,\hat{s})=\left\{
        \begin{aligned}
            0,\;\; \hat{s} = \psi(s),\\
            1,\;\; \hat{s} \neq \psi(s).
        \end{aligned}
        \right.
    \end{align}
\end{definition}
If we apply the deterministic distortion criterion to the weak receiver and set $D_1=0$, we require the weak receiver to reconstruct a function of the source $S$ losslessly. 
\begin{corollary}\label{coro: one helper deterministic distortion capacity region}
    When the deterministic distortion is applied to the weak receiver, the quadruple $(R,0,0,\rho)$ is achievable with a separation-based scheme if and only if
    \begin{align}
        \label{ine: deterministic distortion rate constraint 1}&R > H(\widetilde{S}|V_1) + H(S|\widetilde{S},V_2),\\
        \label{ine: deterministic distortion constraint 2}&I(V_1;T) \leq \rho I(W;Z),\\
        \label{ine: deterministic distortion rate constraint 3}&I(V_1;T) + I(V_2;T|\widetilde{S},V_1)\leq \rho I(W;Z)+ \rho I(X;Y|W).
    \end{align}
with the joint distribution $P_{ST}P_{V_2|T}P_{V_1|V_2}P_WP_{X|W}P_{YZ|X}.$
\end{corollary}
The proof is given in Section \ref{app: proof of corollary one helper deterministic distortion capacity region}.
The separation-based scheme is optimal when $T=S$ and the deterministic distortion measure is applied to the weak receiver with $D_1=0.$
\begin{corollary}
    When $T=S$ with the deterministic distortion measure at the weak receiver, a quadruple $(R,0,D_2,\rho)$ is achievable if and only if
    \begin{align}
        &R + \rho I(W;Z) \geq H(\widetilde{S}),\\
        &R + \rho I(W;Z) + \rho I(X;Y|W) \geq H(\widetilde{S}) + \bar{R}_{S|\widetilde{S}}(D_2),
    \end{align}
    where $\bar{R}_{S|\widetilde{S}}(D_2)=\min_{P_{\hat{S}|S\widetilde{S}}} I(\hat{S};S|\widetilde{S})$ such that $\mathbb{E}[d(S,\hat{S})]\leq D_2$ is the conditional rate--distortion function.
\end{corollary}
\begin{proof}
For the achievability, use Corollary \ref{coro: T=S separation based scheme optimal case} by setting $\hat{S}_1=\widetilde{S}$ and choosing $\hat{S}_2$ such that $I(\hat{S}_2;S|\widetilde{S})=R_{S|\widetilde{S}}(D_2)$. For the converse, invoking Theorem \ref{coro: general outer bound degraded case} by setting $\hat{S}_1=U=\widetilde{S},T=S$, it gives us the bounds
\begin{align*}
    &R + \rho I(W;Z) \geq H(\widetilde{S}),\\
    &R + \rho (I(W;Z) + I(X;Y|W)) \geq H(\widetilde{S}) + I(\hat{S}_2;S|\widetilde{S})\\
    &\quad\quad\quad\quad\quad\quad\quad\quad\quad\quad\quad\;\; \geq H(\widetilde{S}) + \bar{R}_{S|\widetilde{S}}(D_2).\notag \qedhere
\end{align*}

\section{Quadratic Gaussian Case when $\rho=1$}
In this section, we compare the performance of the uncoded and separation schemes for the quadratic Gaussian source when $\rho=1$.

Consider the case that $S\sim\mathcal{N}(0,\sigma_s^2)$ and $T=S$. Furthermore, the channel is a degraded broadcast channel with additive Gaussian noise:
\begin{align}
    &Y = X + N_2,\;\;Z = Y + N_1,
\end{align}
where $N_1 \sim \mathcal{N}(0,\sigma_1^2),N_2 \sim \mathcal{N}(0,\sigma_2^2)$ are independent Gaussian random variables. Without loss of generality, we always assume $D_2 \leq D_1 \leq \sigma_s^2.$ When an uncoded scheme is applied to the helper, the problem reduces to source coding with degraded side information, and the corresponding rate--distortion function was studied in \cite{tian2007multistage}\cite{heegard1985rate}. 

Define 
\begin{align}
    &D_1^* = \frac{\sigma_s^2(\sigma_1^2+\sigma_2^2)}{\sigma_s^2+\sigma_1^2+\sigma_2^2},\;\;D_2^*=\frac{\sigma_s^2\sigma_2^2}{\sigma_s^2+\sigma_2^2},\\
    &\gamma = \frac{\sigma_2^2}{\sigma_1^2+\sigma_2^2},\widetilde{D}_1^* = \frac{\gamma \sigma_1^2D_2}{\gamma \sigma_1^2-(1-\gamma)^2D_2}.
\end{align}
It was shown in \cite{tian2007multistage} that there are four different cases according to different values of $D_1$ and $D_2$:

    \noindent\emph{1).} $\widetilde{D}_1^* \leq D_1 \leq D_1^*$ and $0<D_2\leq D_2^*$: This is the only non-degenerate case and the rate--distortion function is
    \begin{align*}
        R_U(D_1,D_2) = \frac{1}{2}\log \frac{\sigma_s^2\sigma_1^2\sigma_2^2}{D_2(\sigma_s^2+\sigma_1^2+\sigma_2^2)((1-\gamma)^2D_1+\gamma\sigma_1^2)}
    \end{align*}
    \noindent\emph{2).} $D_1 > D_1^*$ and $0 < D_2 \leq D_2^*:$ In this case, encoding for the weak receiver is not necessary and
    \begin{align}
        R_U(D_1,D_2) = R_{S|Y}(D_2),
    \end{align}
    where $R_{S|Y}(D)$ is the Wyner-Ziv rate--distortion function with side information $Y=S+N_2$.
    
    \noindent\emph{3).} $D_1 \leq D_1^*$ and $D_1 \leq \widetilde{D}_1^*:$ In this case, encoding for the strong receiver is not necessary and
    \begin{align}
        R_U(D_1,D_2) = R_{S|Z}(D_1),
    \end{align}
    where $R_{S|Z}(D)$ is the Wyner-Ziv rate--distortion function with side information $Z=S+N_1+N_2$.
    
    \noindent\emph{4).} $D_1 > D_1^*$ and $D_2 > D_2^*$: $R_U(D_1,D_2)=0$.

We next show the rate--distortion function for the separation-based scheme. When $T=S$, invoking Corollary \ref{coro: T=S separation based scheme optimal case} gives
\begin{align}
    &R+ \frac{1}{2}\log \frac{P+\sigma_1^2+\sigma_2^2}{\alpha P + \sigma_1^2+\sigma_2^2} \overset{(a)}{\geq} R + I(W;Z)\overset{(b)}{\geq} R_S(D_1),\\
    &R+ \frac{1}{2}\log \frac{P+\sigma_1^2+\sigma_2^2}{\alpha P + \sigma_1^2+\sigma_2^2} + \frac{1}{2}\log \frac{\alpha P + \sigma_2^2}{\sigma_2^2}\notag \\
    &\overset{(c)}{\geq} R + I(W;Z) + I(X;Y|W)\overset{(d)}{\geq} R_S(D_2),
\end{align}
for some $\alpha\in[0,1],$ where $(a)$ and $(c)$ follow from the converse of Gaussian degraded broadcast channels \cite[Section 5.5.2]{el2011network}, $(b)$ and $(d)$ follow from the definition of the rate--distortion function. The achievability follows by setting $W\sim\mathcal{N}(0,(1-\alpha)P), V\sim \mathcal{N}(0,\alpha P)$ and $X=W + V$ for the channel coding and the cascade of Gaussian test channels
\begin{align}
    S = \hat{S}_2 + U_2 = (\hat{S}_1 + U_1) + U_2,
\end{align}
where $\hat{S}_1\sim \mathcal{N}(0,\sigma_s^2 - D_1),U_1\sim\mathcal{N}(0,D_1-D_2),U_2\sim\mathcal{N}(0,D_2).$ Here we use the fact that a Gaussian source with mean squared error is successively refinable. The corresponding rate--distortion function in this case is
\begin{align}
    R_{Se}(D_1,D_2) = \min_{\alpha\in[0,1]}\max\left\{\frac{1}{2}\log \frac{\sigma_s^2(\alpha P + \sigma_1^2+\sigma_2^2)}{D_1(P+\sigma_1^2+\sigma_2^2)},\right.\\
    \left.\frac{1}{2}\log \frac{\sigma_s^2\sigma_2^2(\alpha P + \sigma_1^2+\sigma_2^2)}{D_2(P+\sigma_1^2+\sigma_2^2)(\alpha P + \sigma_2^2)}\right\}.
\end{align}
We further restrict $P=\sigma_s^2$, so that the two schemes have the same input power. Denote the two terms in $\max \{\cdot\}$ by $R_1(\alpha)$ and $R_2(\alpha)$, respectively. Note that $R_1$ is an increasing function of $\alpha$ and $R_2$ is a decreasing function of $\alpha$. Furthermore, when $D_1 \frac{\sigma_2^2}{\sigma_s^2+\sigma_2^2}\leq D_2$, there exists an $\alpha\in[0,1]$ such that $R_1(\alpha)=R_2(\alpha)$ and the minimum of $R_{Se}(D_1,D_2)$ is achieved. Otherwise, $R_2(\alpha)\geq R_1(\alpha)$ always hold. Denote $\widetilde{D}_2^*:=D_1 \frac{\sigma_2^2}{\sigma_s^2+\sigma_2^2}$. In the following, we compare the rate--distortion function with the separation-based scheme and the uncoded scheme.

    \noindent\emph{1).} $\widetilde{D}_1^* \leq D_1 \leq D_1^*$ and $\widetilde{D}_2^*<D_2\leq D_2^*$: In this case, $R_{Se}(D_1,D_2)$ achieves its minimum when $\alpha = \frac{\sigma_2^2}{\sigma_s^2}\frac{D_1-D_2}{D_2}$ and
    \begin{align}\label{ine: gaussian separation rate for nondegenerate case}
        R_{Se}(D_1,D_2) = \frac{1}{2}\log \frac{\sigma_s^2(\sigma_2^2 D_1 + \sigma_1^2 D_2)}{D_1D_2(\sigma_s^2+\sigma_1^2+\sigma_2^2)}.
    \end{align}
    One can easily verify that $R_{Se}(D_1,D_2)\leq R_U(D_1,D_2)$ if and only if
    \begin{align}
        \widetilde{D}_2^*\leq D_2 \leq \frac{D_1\sigma_2^2(\sigma_1^2+\sigma_2^2-D_1)}{\sigma_1^2D_1+\sigma_2^2(\sigma_1^2+\sigma_2^2)}.
    \end{align}
    Such $D_2$ exists when $D_1 \leq D_1^*.$
    
    \noindent\emph{2).} $\widetilde{D}_1^* \leq D_1 \leq D_1^*$ and $D_2\leq \widetilde{D}_2^*< D_2^*$: In this case, $R_2(\alpha)\geq R_1(\alpha)$ always holds and the minimum of $R_{Se}(D_1,D_2)$ is achieved when $\alpha = 1$, which implies $R_{Se}(D_1,D_2) = R_{S|Y}(D_2)$. By the assumption that $D_1 \leq D_1^*,$ it follows that $R_{Se}(D_1,D_2)\leq R_U(D_1,D_2).$
    
    \noindent\emph{3).} $D_1 > D_1^*$ and $0 < D_2 \leq D_2^*:$ As $R_{Se}(D_1,D_2)\geq R_{S|Y}(D_2)$ always holds, in this case we have $R_{Se}(D_1,D_2)\geq R_U(D_1,D_2)$.
    
     \noindent\emph{4).} $D_1 \leq D_1^*$ and $D_1 \leq \widetilde{D}_1^*$ and $\widetilde{D}_2^*<D_2\leq D_2^*$: In this case the uncoded rate $R_U(D_1,D_2)=R_{S|Z}(D_1)$ and the separation-based rate is \eqref{ine: gaussian separation rate for nondegenerate case}. It follows that
     \begin{align}
            &R_{Se}(D_1,D_2) \geq R_U(D_1,D_2)\\
         \Leftrightarrow &\frac{\sigma_s^2(\sigma_2^2 D_1 + \sigma_1^2 D_2)}{D_1D_2(\sigma_s^2+\sigma_1^2+\sigma_2^2)} \geq \frac{\sigma_s^2(\sigma_1^2+\sigma_2^2)}{D_1(\sigma_s^2+\sigma_1^2+\sigma_2^2)}\\
         \Leftrightarrow &\frac{\sigma_2^2 D_1 + \sigma_1^2 D_2}{D_2} \geq \sigma_1^2+\sigma_2^2\\
         \Leftrightarrow & \sigma_1^2 + \sigma_2^2 \frac{D_1}{D_2} \geq \sigma_1^2 + \sigma_2^2.
     \end{align}
     \noindent\emph{5).} $D_1 \leq D_1^*$ and $D_1 \leq \widetilde{D}_1^*$ and $D_2\leq \widetilde{D}_2^*\leq D_2^*$: In this case, by $D_2 \leq \widetilde{D}_2^*$ we have a lower bound on $D_1:$
     \begin{align}
         D_1 \geq D_2 \frac{\sigma_s^2+\sigma_2^2}{\sigma_2^2}.
     \end{align}
     Combining the upper bounds $D_1\leq D_1^*$ and $D_1\leq \widetilde{D}_1^*$ gives
     \begin{align}
        &D_2 \frac{\sigma_s^2+\sigma_2^2}{\sigma_2^2} \leq D_1\leq \frac{\sigma_s^2(\sigma_1^2+\sigma_2^2)}{\sigma_s^2+\sigma_1^2+\sigma_2^2}\\
        \Leftrightarrow & D_2 \leq \frac{\sigma_s^2\sigma_2^2(\sigma_1^2+\sigma_2^2)}{(\sigma_s^2+\sigma_1^2+\sigma_2^2)(\sigma_1^2+\sigma_2^2)}
     \end{align}
     and
     \begin{align}
         &D_2 \frac{\sigma_s^2+\sigma_2^2}{\sigma_2^2} \leq D_1\leq\frac{\gamma \sigma_1^2D_2}{\gamma \sigma_1^2-(1-\gamma)^2D_2}\\
         \Leftrightarrow & D_2 \geq \frac{\sigma_s^2\sigma_2^2(\sigma_1^2+\sigma_2^2)}{\sigma_1^2(\sigma_s^2+\sigma_2^2)},
     \end{align}
     which is a contradiction since $\sigma_1^2(\sigma_s^2+\sigma_2^2)\leq (\sigma_s^2+\sigma_1^2+\sigma_2^2)(\sigma_1^2+\sigma_2^2).$

\begin{figure}
    \centering
    \includegraphics[scale=0.5]{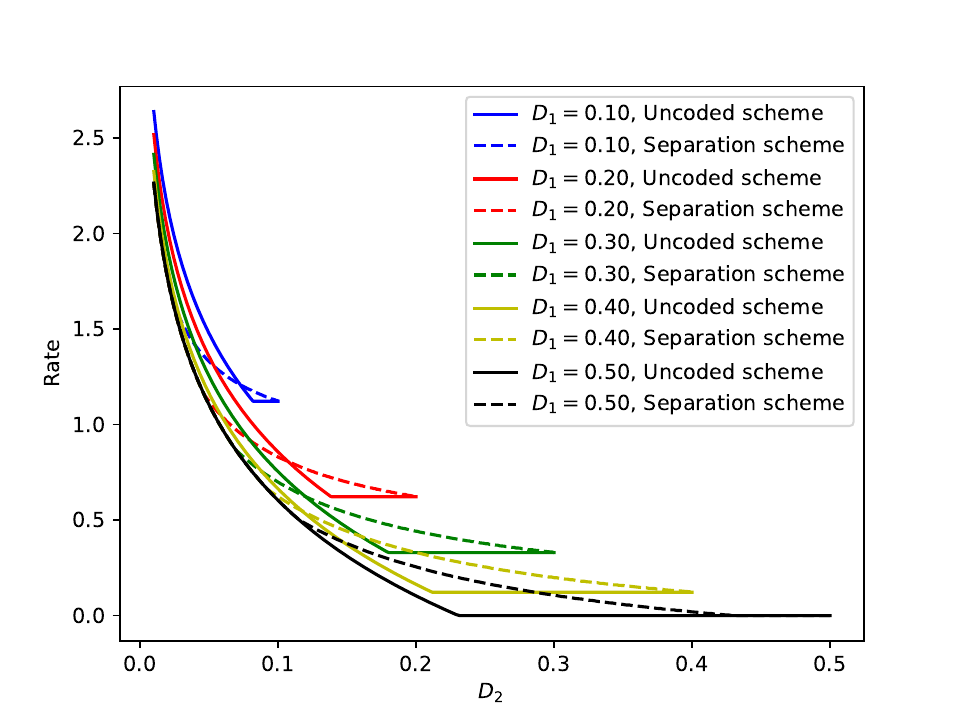}
    \caption{\footnotesize Comparison between the uncoded scheme and the separation-based scheme for different values of $D_1$ when $\sigma_s^2=1,\sigma_1^2=0.6,\sigma_2^2=0.3$.}
    \label{fig:gaussian example}
\end{figure}
Fig. \ref{fig:gaussian example} shows the rate--distortion function $R(D_2)$ of the uncoded scheme and the separation-based scheme given different values of $D_1$. 
\end{proof}

\section{proof of Theorem \ref{coro: general outer bound degraded case}}\label{app: proof of corollary general outer bound degraded case}
Fix a joint distribution $P_{ST}P_{U|T}$ and generate an auxiliary sequence $U^n$ by $P_{U|T}^n=\prod_{i=1}^n P_{U|T}$. It follows that
    \begin{align*}
        &nR \geq H(M)\\
        &\geq I(M;U^n,S^n|Z^l)\\
        &=I(M;U^n|Z^l) + I(M;S^n|Z^l,U^n)\\
        &\overset{(a)}{\geq} I(M;U^n|Z^l) + I(M;S^n|Z^l,U^n,Y^l)\\
        &=I(M,Z^l;U^n) - I(U^n;Z^l) + I(M,Z^l,Y^l;S^n|U^n) \\
        &\quad\quad\quad\quad\quad\quad\quad\quad - I(Z^l,Y^l;S^n|U^n)\\
        &\geq I(\hat{S}^n_1;U^n) - I(U^n;Z^l) + I(\hat{S}^n_1,\hat{S}^n_2;S^n|U^n)\\
        &\quad\quad\quad\quad\quad\quad\quad\quad- I(Z^l,Y^l;S^n|U^n)\\
        &\geq \sum_{i=1}^n I(\hat{S}_{1,i};U_i) - I(U_i;Z^l,U^{i-1}) + I(\hat{S}_{1,i},\hat{S}_{2,i};S_i|U_i)\\
        &\quad\quad\quad\quad\quad\quad\quad\quad - I(Z^l,Y^l,S^{i-1},T^{i-1},U^{n\backslash i};S_i|U_i)\\
        &\overset{(b)}{\geq}n(I(\hat{S}_1;U) + I(\hat{S}_1,\hat{S}_2;S|U) - I(U;V_1) - I(V_1,V_2;S|U))
    \end{align*}
    where $(a)$ follows from the fact that condition reduces entropy and the Markov chain $M-S^n-(Z^l,U^n,Y^l)$, $(b)$ follows by introducing a time-sharing random variable $Q$ and setting $\hat{S}=\hat{S}_Q,V_{1,Q} = (Z^l,U^{Q-1},Q),V_1 = V_{1,Q},V_2=V_{2,Q}=(Y^l,S^{Q-1},T^{Q-1},U^{n\backslash Q},Q)$.
    
    For the constraints of the helper,
    \begin{align}
        &l I(W;Z)=\sum_{i=1}^l I(U^n,Z^{n}_{i+1},Y^{i-1};Z_i)\\
        &\geq \sum_{i=1}^l I(U^n,Z^{n}_{i+1};Z_i)\\
        &\geq I(U^n;Z^l) = \sum_{i=1}^n I(U_i;Z^l,U^{i-1})=n I(U;V_1),
    \end{align}
    where the first equality follows by introducing a time-sharing random variable $\widetilde{Q}$ and setting $W=W_{\widetilde{Q}}= (U^n,Z^{n}_{\widetilde{Q}+1},Y^{\widetilde{Q}-1},\widetilde{Q})$. Next,
    \begin{align}
        &I(U^n;Z^l) + I(T^n;Y^l|U^n) \\
        &\overset{(a)}{=} \sum_{i=1}^n I(U_i;Z^l,U^{i-1}) + I(T_i;Y^l|U^n,T^{i-1},S^{i-1})\\
        &=\sum_{i=1}^n I(U_i;Z^l,U^{i-1}) + I(T_i;Y^l,T^{i-1},U^{n \backslash i},S^{i-1}|U_i)\\
        &\geq n I(U;V_1) + n I(T;V_2|U),
    \end{align}
    where $(a)$ follows from the Markov chain $S^{i-1}-(T^{i-1},Y^l,U^n)-T_i$.
    The upper bound proof of $I(U^n;Z^l) + I(T^n;Y^l|U^n)$ is similar to the converse of more capable broadcast channels and is omitted here.

\section{proof of corollary \ref{coro: T=S separation based scheme optimal case}}\label{app: proof of T=S separation based scheme optimal case}
\emph{Converse: } For the first bound,
\begin{align}
    R + H(M_1) &\geq I(M;S^n|M_1) + I(M_1;S^n) \\
    &= I(M,M_1;S^n) \geq I(\hat{S}^n_1;S^n)\\
    &=\sum_{i=1}^n I(\hat{S}^n_1,S^{i-1};S_i)\\
    &\geq \sum_{i=1}^n I(\hat{S}_{1i};S_i)\geq n I(\hat{S}_1,S),
\end{align}
where the last step follows by introducing a time-sharing random variable.
For the second bound,
\begin{align}
    &R + H(M_1,M_2)\geq I(M;S^n|M_1,M_2) + I(M_1,M_2;S^n)\\
    &=I(M,M_1,M_2;S^n) \geq I(\hat{S}^n_1,\hat{S}^n_2;S^n) \geq n I(\hat{S}_1,\hat{S}_2;S).
\end{align}
The upper bounds follow by using Fano's inequality and the converse for the more capable broadcast channel \cite{el2011network} and are omitted here.

\section{proof of corollary \ref{coro: one helper deterministic distortion capacity region}}\label{app: proof of corollary one helper deterministic distortion capacity region}
The achievability follows by setting $U_1=\widetilde{S},U_2=S$ in Theorem \ref{the: general one helper inner bound}. To see the converse, define
\begin{align*}
    &V_{1,i} = (M_1,\widetilde{S}^{i-1}),\; V_{1,Q} = (M_1,\widetilde{S}^{Q-1},Q)\\
    &V_{2,i} = (M_2,V_{1,i},\widetilde{S}^n_{i+1},T^{i-1}),\;V_{2,Q} = (M_2,V_{1,Q},\widetilde{S}^n_{Q+1},T^{Q-1}).
\end{align*} 
We have
\begin{align}
    &H(M_1) \geq I(M_1;T^n)\notag \\
    &=\sum_{i=1}^n I(M_1,\widetilde{S}^n;T^n)-I(\widetilde{S}^n;T^n|M_1)\notag\\
    &=\sum_{i=1}^n I(M_1,\widetilde{S}^n;T_i|T^{i-1}) - I(\widetilde{S}_i;T^n|M_1,\widetilde{S}^{i-1})\notag\\
    &\overset{(a)}{=}\sum_{i=1}^n I(M_1,\widetilde{S}^n,T^{i-1};T_i) - I(\widetilde{S}_i;T_i|M_1,\widetilde{S}^{i-1})\notag\\
    &=\sum_{i=1}^n I(M_1,\widetilde{S}^{i-1};T_i) +I(\widetilde{S}^{n}_{i},T^{i-1};T_i|M_1,\widetilde{S}^{i-1}) \notag\\
    &\quad\quad\quad\quad - I(\widetilde{S}_i;T_i|M_1,\widetilde{S}^{i-1})\notag\\
    &\label{ine: deterministic distortion converse ine 1}=\sum_{i=1}^n I(M_1,\widetilde{S}^{i-1};T_i) +I(\widetilde{S}^{n}_{i+1},T^{i-1};T_i|M_1,\widetilde{S}^{i-1},\widetilde{S}_i)\\
    &\geq \sum_{i=1}^n I(M_1,\widetilde{S}^{i-1};T_i)\notag
\end{align}
where $(a)$ is by the Markov chain $\widetilde{S}_i-T^{n\backslash i}-(T_i,M_1,\widetilde{S}^{i-1})$, and
\begin{align*}
    &H(M_1) + H(M_2|M_1)\\
    &\overset{(a)}{\geq} \sum_{i=1}^n I(M_1,\widetilde{S}^{i-1};T_i) +I(\widetilde{S}^{n}_{i+1},T^{i-1};T_i|M_1,\widetilde{S}^{i-1},\widetilde{S}_i)\\
    &\quad\quad\quad\quad\quad\quad\quad\quad + I(T^n;M_2|\widetilde{S}^n,M_1)\\
    &=\sum_{i=1}^n I(M_1,\widetilde{S}^{i-1};T_i) +I(\widetilde{S}^{n}_{i+1},T^{i-1};T_i|M_1,\widetilde{S}^{i-1},\widetilde{S}_i)\\
    &\quad\quad\quad\quad\quad\quad\quad\quad + I(T_i;M_2|M_1,\widetilde{S}^{i-1},\widetilde{S}_i,\widetilde{S}_{i+1}^n,T^{i-1})\\
    &=\sum_{i=1}^n I(M_1,\widetilde{S}^{i-1};T_i) +I(\widetilde{S}^{n}_{i+1},T^{i-1},M_2;T_i|M_1,\widetilde{S}^{i-1},\widetilde{S}_i)
\end{align*}
where $(a)$ follows from \eqref{ine: deterministic distortion converse ine 1}. The upper bounds of $H(M_1)$ and $H(M_1,M_2)$ are the same as the converse of more capable broadcast channels.
Finally, it follows that
\begin{align*}
    nR &\geq H(M)\\
    &\geq I(M;S^n,\widetilde{S}^n,M_2|M_1)\\
    &= I(M;\widetilde{S}^n|M_1) + I(M;S^n,M_2|\widetilde{S}^n,M_1)\\
    &\overset{(a)}{\geq} H(\widetilde{S}^n|M_1) - n\epsilon + I(M;S^n|\widetilde{S}^n,M_1,M_2)\\
    &\overset{(b)}{\geq} H(\widetilde{S}^n|M_1) + H(S^n|\widetilde{S}^n,M_1,M_2) - 2n\epsilon\\
    &=\sum_{i=1}^n H(\widetilde{S}_i|\widetilde{S}^{i-1},M_1) + H(S_i|\widetilde{S}^n,S^{i-1},M_1,M_2) - 2n\epsilon\\
    &\geq \sum_{i=1}^n H(\widetilde{S}_i|\widetilde{S}^{i-1},M_1) + H(S_i|\widetilde{S}^n,S^{i-1},M_1,M_2,T^{i-1}) - 2n\epsilon\\
    &\overset{(c)}{=}\sum_{i=1}^n H(\widetilde{S}_i|\widetilde{S}^{i-1},M_1) + H(S_i|\widetilde{S}^n,M_1,M_2,T^{i-1}) - 2n\epsilon
\end{align*}
where $(a)$ and $(b)$ follows from Fano's inquality, $(c)$ follows from the Markov chain $S_i-(\widetilde{S}^n,M_1,M_2,T^{i-1})-S^{i-1}$. Introducing a time-sharing random variable completes the proof.

\clearpage
\newpage
\bibliographystyle{ieeetr} 
\bibliography{ref}

\clearpage
\newpage
\onecolumn

\appendices

\section{One-Receiver case}\label{app: single user outer bound}
\begin{proposition}\label{prop: single receiver necessary condition}
    Given a joint distribution $P_{ST}$, for the lossy compression of the source $S$ with a helper who observes $T$ and communicates to the receiver through a DMC $P_{Y|X},$ a rate-distortion-bandwidth triple $(R,D,\rho)$ is achievable only if
    \begin{align}
        &R \geq R_S(D) - I(V;S),\\
        &\rho C \geq  I(T;V)
    \end{align}
    for some $S-T-V$, where $R_S(D)$ is the rate-distortion function such that $R_S(D) = \min_{P_{\hat{S}|S}}I(S;\hat{S}), \mathbb{E}[d(S,\hat{S})]\leq D$, $C=\max_{P_X} I(X;Y)$ is the capacity of the channel $P_{Y|X}.$
\end{proposition}

\begin{remark}
    The outer bound is tight in the following cases:
    \begin{itemize}
        \item When $S$ is required to be reconstructed losslessly: In this case, the separation-based scheme at the helper is optimal. The outer bound reduces to
        \begin{align}
            &R \geq H(S|V),\;\rho C \geq I(T;V).
        \end{align}
        The helper can always choose $V$ as a quantized version of $T$ such that $I(T;V)$ is slightly smaller than $\rho C$.
        \item When $T$ is a function of $S$, we have $R+\rho C \geq R_S(D) + I(T;V|S)=R_S(D).$
        \item Quadratic Gaussian source with additive Gaussian channel: Suppose $(S,T)$ is a pair of correlated Gaussian random variables with zero mean, unit variance, and correlation $\nu$. The channel is an additive Gaussian channel with noise variance $\sigma_N^2$ and input power constraint $P$. In this case, we have
        \begin{align}
            &nR \geq I(S^n;\hat{S}^n) - I(S^n;Y^l)\\
            &\geq H(S^n) - H(S^n|\hat{S}^n) - H(S^n) + H(S^n|Y^l)\\
            &\geq  \frac{n}{2}\log  \frac{1}{2\pi e D} + H(S^n|Y^l).
        \end{align}
        It follows that $S = \nu T + W$, where $W\sim \mathcal{N}(0,1-\nu^2)$ and
        \begin{align}
            2^{2H(S^n|Y^l)/n} &\geq 2^{2H(\nu T^n|Y^l)/n} + 2^{2H(W^n)/n}\\
            &=\nu^2 2^{2H(T^n|Y^l)/n} + 2\pi e(1-\nu^2).
        \end{align}
        Furthermore, we have
        \begin{align}
            l C \geq I(T^n;Y^l) = \frac{n}{2}\log 2\pi e - H(T^n|Y^l),
        \end{align}
        which gives
        \begin{align}
            \frac{1}{n}H(S^n|Y^l) \geq \frac{1}{2}\log 2\pi e (1-\nu^2 + \nu^2 2^{-2\rho C})
        \end{align}
        and
        \begin{align}
            R \geq \frac{1}{2}\log \frac{1-\nu^2 + \nu^2 2^{-2\rho C}}{D}.
        \end{align}
        For the achievability, it is straightforward to see an achievability region for the problem is
        \begin{align}
            &R \geq I(U;S|V),\\
            &\rho C > I(V;T)
        \end{align}
        with the Markov chain $U-S-T-V$ such that $\mathbb{E}[d(S,h(U,V))]\leq D$ for some function $h:\mathcal{U}\times\mathcal{V}\to \hat{S}$. One can choose test channels $U=S+N_1,V=T+N_2$ where $N_1\sim \mathcal{N}(0,\sigma_{N_{1}}^2),N_2\sim \mathcal{N}(0,\sigma_{N_{2}}^2)$ such that 
        \begin{align}
            I(V;T)=\frac{1}{2}\log \frac{1+\sigma_{N_2}^2}{\sigma_{N_2}^2} < \rho C.
        \end{align}
        Following the argument in \cite{wagner2008rate}\cite[Section 12.3]{el2011network}, the rate is
        \begin{align}
            R \geq \frac{1}{2}\log \frac{1-\nu^2 + \nu^2 2^{-2 I(V;T)}}{D}\geq \frac{1}{2}\log \frac{1-\nu^2 + \nu^2 2^{-2 \rho C}}{D}.
        \end{align}
    \end{itemize}
\end{remark}
The value of the rate-distortion $R_S(D)$ is the rate needed to achieve the distortion constraint, and $I(V;S)$ is the reduced rate at the source sender due to the message from the helper. However, for the helper, there is an additional rate penalty term $I(T;V|S)$ by its noisy observation $T$.
\begin{proof}[Proof of Proposition \ref{prop: single receiver necessary condition}]
    For the sender rate:
    \begin{align}
        nR &\geq H(M)\\
        &\geq I(M;S^n|Y^l)\\
        &=I(M,Y^l;S^n) - I(Y^l;S^n)\\
        &=I(M,Y^l,\hat{S}^n;S^n) - I(Y^l;S^n)\\
        &=\sum_{i=1}^n I(M,Y^l,\hat{S}^n;S_i|S^{i-1}) - I(Y^l;S_i|S^{i-1})\\
        &=\sum_{i=1}^n I(M,Y^l,\hat{S}^n,S^{i-1};S_i) - I(Y^l,S^{i-1};S_i)\\
        &\geq \sum_{i=1}^n I(S_i;\hat{S}_i) - I(Y^l,S^{i-1},T^{i-1};S_i)\\
        &\geq nR_S(D) - nI(V;S),
    \end{align}
    where the last equality follows by introducing a time-sharing random variable $Q$ and setting $V_Q = (Y^l,S^{Q-1},T^{Q-1},Q)$ and $V=V_Q,S=S_Q$. For the sum rate,
    \begin{align}
        &I(T^n;Y^l)\\
        &=\sum_{i=1}^n I(T_i;Y^l|T^{i-1})\\
        &=\sum_{i=1}^n H(T_i|T^{i-1}) - H(T_i|T^{i-1},Y^l)\\
        &\overset{(a)}{=} H(T_i|T^{i-1},S^{i-1}) - H(T_i|T^{i-1},Y^l,S^{i-1})\\
        &=\sum_{i=1}^n I(T_i;Y^l|T^{i-1},S^{i-1})\\
        &=\sum_{i=1}^n I(T_i;Y^l,T^{i-1},S^{i-1})\\
        &=n I(T;V),
    \end{align}
    where $(a)$ follows by the i.i.d. property of $(S_i,T_i),i=1,,,n$ and the Markov chain $S^{i-1}-(T^{i-1},Y^l)-T_i$. The upper bound follows from
    \begin{align}
        I(T^n;Y^l)\leq \sum_{i=1}^l I(T^n,Y^{i-1};Y_i) \leq l C.
    \end{align}
    The proof is completed.
\end{proof}

\section{achievability of corollary \ref{coro: T=S separation based scheme optimal case}}\label{app: achievability of T=S separation case}
 Fix a joint distribution $P_{\hat{S}_1\hat{S}_2|S}$ and $\epsilon>0$ such that
\begin{align}
    &\mathbb{E}[d(S,\hat{S}_1)]\leq D_1,\mathbb{E}[d(S,\hat{S}_2)]\leq D_2
\end{align}
Generate a codebook $\{\hat{S}^n_1\}$ of size $2^{n(R_{11}+R_{21})}$ such that each codeword is i.i.d. generated according to the distribution $P_{\hat{S}_1|S}$.
The indices of all $\hat{S}^n_1$ are split into two independent sets $[1:2^{nR_{11}}]$ and $[1:2^{nR_{21}}]$.

For each $\hat{s}^n_1,$ generate a codebook $\{\hat{s}^n_2|\hat{s}^n_1\}$ of size $2^{n (R_{12}+R_{22})}$ such that each codeword is i.i.d. generated according to the distribution $P_{\hat{S}_2|\hat{S}_1}$.  Again, the indices of $\hat{S}^n_2$ in each codebook are split into two independent sets $[1:2^{nR_{12}}]$ and $[1:2^{nR_{22}}]$.

Upon observing the source $s^n$, the sender and the helper look for a $u^n_1$ such that $(\hat{s}^n_1,s^n)\in\mathcal{T}^n_{P_{S\hat{S}},\delta}$. If there is more than one such codeword, select the first one. If there is no such codeword, the sender declares an error. The selected $\hat{s}^n_1$ is indexed by $(M_{11},M_{21})$. Then the sender and the helper look for a codeword $\hat{s}^n_2$ in $\{\hat{s}^n_2|\hat{s}^n_1\}$ such that $(\hat{s}^n_1,\hat{s}^n_2,s^n)\in \mathcal{T}^n_{P_{\hat{S}_1\hat{S}_2S},\delta}$. If there is more than one such codeword, select the first one. If there is no such codeword, the sender declares an error. The selected codeword can be indexed by $(M_{12},M_{22})$. Then the sender sends $(M_{11},M_{12})$ through the noiseless link and the helper uses a degraded BC code to send $(M_{21},M_{22})$ with $M_{21}$ being the common message and $M_{22}$ being the private message. Coding error probability vanishes with sufficiently large $n$ if the following inequalities and equalities hold:
\begin{align}
    &R_{11} + R_{21} \geq I(\hat{S}_1;S),\;\;R_{12} + R_{22} \geq I(\hat{S}_2;S|\hat{S}_1),\\
    &R_{21} \leq I(W;Z),\;\;R_{21} + R_{22} \leq I(W;Z) + I(X;Y|W),\\
    &R=R_{11} + R_{12}, R_{11}\geq 0, R_{12}\geq 0, R_{21}\geq 0, R_{22}\geq 0.
\end{align}
Applying the Fourier-Motzkin elimination yields
\begin{align}
    &R + I(W;Z) \geq I(\hat{S}_1;S),\\
    &R+ I(W;Z) + I(X;Y|W) \geq I(\hat{S}_1,\hat{S}_2;S).
\end{align}

\end{document}